\documentclass{article}
\usepackage{macros}

\title{Different Forms of Imbalance in Strongly Playable Discrete Games I: Two-Player $RPS$ Games}
\author{Itai Maimon}
\affil{Department of Mathematics \\
University of California San Diego \\
La Jolla, CA 92093-0112 (USA) \\
\{\tt imaimon\}@ucsd.edu}
\date{\today}

\begin{document}

\maketitle
\begin{abstract}
We construct several definitions of imbalance and playability, both of which are related to the existence of dominated strategies. Specifically, a maximally balanced game and a playable game cannot have dominated strategies for any player. In this context, imbalance acts as a measure of inequality in strategy, similar to measures of inequality in wealth or population dynamics. Conversely, playability is a slight strengthening of the condition that a game has no dominated strategies. It is more accurately aligned with the intuition that all strategies should \textit{see play}. We show that these balance definitions are natural by exhibiting a $(2n+1)$-$RPS$ that maximizes all proposed imbalance definitions among playable $RPS$ games. We demonstrate here that this form of imbalance aligns with the prevailing notion that different definitions of inequality for economic and game-theoretic distributions must agree on both the maximal and minimal cases \cite{Cowell2000_MeasurementOfInequality}. In the sequel paper, we utilize these definitions for multiplayer games to demonstrate that a generalization of this imbalanced $RPS$ is at least nearly maximally imbalanced while remaining playable for under $50$ players.
$\,$

$\,$

\end{abstract}

\tableofcontents
\newpage

\section{Introduction: Definitions for Imbalance and Playability}
Classic Rock-Paper-Scissors ($RPS$) has seen many variants, generalizations, and applications in the past several years \cite{KerrEtAl2002LocalDispersalRPS} \cite{SpiroSuryaZeng2022_SemiRestrictedRPS}. An example application of this game appears in the population dynamics of the males of the side-blotched lizard \cite{SinervoLively1996RPSLizards}.  There are three varieties of these male lizards differentiated by their distinctive orange, blue, or yellow throat coloring. The $RPS$ dynamic appears in the mating proportions of the species, where the orange males have a direct mating advantage over the blues and, in isolation, would overpopulate them to extinction. However, as soon as there is a majority of orange lizards, the growth rate of the yellow lizards increases, as they have a mating advantage over the orange lizards. Finally, as the population of yellow lizards grows and the number of orange lizards decreases, the growth rate of blue lizards increases because of a similar mating advantage, and the cycle continues. 

Another example application occurs in the world of competitive trading card games. In competition, most strategies are describable in terms of the archetypes of deck constructions, which have a fundamental $RPS$ structure between them. For instance, decks that start very aggressively generally lose against decks that build up a little slower and more linearly, which typically lose against even slower decks that need time to build to a given combo, which then, in turn, lose against the very fast aggressive decks. This $RPS$ dynamic ensures that no pure strategy dominates the competition.

The $RPS$ game at the base of most of these models is primarily balanced in that each choice of object \textit{beats} half of the remaining objects and \textit{is beaten by} the other half. We can then consider what the effect would be if we reversed this approach. Unfortunately, a maximally imbalanced $RPS$ game should contain a single object that beats all other objects, but then the only undominated strategy would be to play just this object. This game clearly cannot be used to model any interesting dynamics, as almost all objects will never be played. To counteract this, we want to consider games that are maximally imbalanced yet playable in that each object \textit{sees play}. Our thoughts on the effect of imbalance on ecologies and trading card games are expanded in section \ref{applications}. Broadly, we would expect that imbalanced playable games model more asymmetrical yet stable environments or meta-games, respectively. We further conjecture that this stable imbalance gives rise to interesting and desirable properties. This work focuses on making these imbalanced definitions sensible within this context by directly showing that, while they may disagree on some instances, they agree on the most and least balanced playable $RPS$ games. Thereby, confirming the commonly held notion that different measures of inequality/asymmetry should agree on the most inequitable and most equitable distributions \cite{Cowell2000_MeasuringInequality}. 

To make these terms more precise, we give the following definitions:
\begin{defi}[$RPS$]
We define an $n$-$RPS$ game as a game played on an unlabeled $n$-tournament, or a choice of orientations on each edge of the complete graph on $n$ vertices, $K_n$. It is played by two players, $p$ and $q$, each of whom chooses a vertex, $u_p,u_q\in K_n$. If $u_p=u_q$, the game is a tie, and both players receive a payoff of $0$. Otherwise, on the edge $u_pu_q$, the player who chose the negatively oriented, or outgoing, vertex loses, and the player who chose the positively oriented, or incoming,  vertex wins. The payoff for the winner is $+1$, and for the loser, $-1$. Each pure strategy in this game corresponds to choosing to play a specific vertex, which in $RPS$ terminology we define as choosing an object (e.g., one chooses an object from rock, paper, or scissors). For an $n$-$RPS$, $G$, we will refer to $G$ as both the tournament that the game is based on and the game itself played on said tournament, e.g., we can talk both of the vertices of $G$ and the payoffs of players $p,q$ in $G$. 
\end{defi}

Playability for a two-player $RPS$ game will correspond broadly to the existence of a Nash equilibrium in which both players' mixed strategies encompass all possible pure strategies. Here we define an extension of this concept that can also be applied to many types of multiplayer games, and captures that each object has a positive probability of being played in an instance of a Nash equilibrium. We will explore these extensions more heavily in the sequel. We will see in lemma \ref{RPS-is-trivial} that in the two-player case, this will force both strategies to be totally mixed.

\begin{defi}[$k$-Playable]
For $k>0$, an $RPS$ game is $k$-playable if a Nash equilibrium, $N$, exists, such that each object has a positive probability of being played by $k$ players. In other words, in $N$, for each object, $o$, there are at least $k$ players, $p^1_o\dots p^k_o$, whose mixed strategy has a positive probability of playing $o$. In this case, we say that $o^k$ \textit{sees play} by $p^1_o\dots p^k_o$ in Nash equilibrium $N$. A game is $k$-strongly playable if each Nash equilibrium has this property. On the other hand, a game is $k$-weakly playable if for each object $o$ there is a Nash equilibrium, $N_o$, in which $o^k$ sees play. We say a game is playable if it is $1$-playable, and similarly define a game as strongly/weakly playable if it is $1$-strongly/$1$-weakly playable.
\end{defi}

\subsubsection*{$k$-Playability for Asymmetric and Non-Discrete Games}
This definition is directly applicable to any symmetric game over a discrete set of strategies. However, for asymmetric games, playability must be defined as the existence of a Nash equilibrium where each player plays a totally mixed strategy. Then strong playability would imply that this is true for each Nash equilibrium, and weak playability that for each pure strategy, there exists a Nash equilibrium in which it sees play for the player that can play it.

For a non-discrete set of strategies, we replace the positive probability condition with the condition that the probability density is positive almost everywhere. For example, in an asymmetric game, if there exists an equilibrium such that each player's strategy has positive probability density almost everywhere, that game is playable; similarly, if all equilibria are playable, it is strongly playable. Finally, for each player $i$, we can consider the intersection over each Nash equilibrium, $N_\alpha$, of the subset for which player $i$'s probability densities in $N_\alpha$ are zero. If that set has measure zero for each player, then the game is weakly playable. We can similarly adjust these definitions for a subset of $k$ players in symmetric non-discrete games.  

\subsubsection*{Strongly Playable implies no Weakly Dominated Strategies and Weakly Playable implies no Dominated Strategies}
In order to better understand these conditions, it is worth checking how the playable conditions interact with the existence of dominated strategies. Before doing so, we should confirmf that the terminology of strong and weak is apt as it follows directly from the definition that a $k$-strongly playable game must also be $k$-playable and $(k-1)$-strongly playable. Similarly, a $k$-playable game is also both $k$-weakly playable and $(k-1)$-playable. Finally, a $k$-weakly playable game is also $(k-1)$-weakly playable. 

We can also show that in a weakly playable game, no strategy can be dominated, and in a strongly playable game, no strategy can be weakly dominated. For instance, if strategy $\sigma'$ dominates $\sigma$ for player $p$, then we can take the difference of these to form strategies $\tilde{\sigma}'$ and $\tilde{\sigma}$. Here, we define the difference of strategy $\sigma$ with respect to $\sigma'$ as the mixed strategy $\tilde{\sigma}$ with relative probability for each pure strategy $i$ versus $j$ of:
\begin{equation}
    \frac{P_{\tilde{\sigma}}(i)}{P_{\tilde{\sigma}}(j)}=\frac{P_\sigma(i)-\min(P_\sigma(i),P_{\sigma'}(i))}{P_\sigma(j)-\min(P_\sigma(j),P_{\sigma'}(j))}
\end{equation}
By construction, $\tilde{\sigma}'$ and $\tilde{\sigma}$ do not contain any of the same pure strategies, yet $\tilde{\sigma}'$ must still dominate $\tilde{\sigma}$. We can then see that any pure strategy in $\tilde{\sigma}$ does not see play in any Nash equilibrium, and the game cannot be weakly playable. Following identical logic shows that if $\sigma$ is weakly dominated by $\sigma'$, we can construct a Nash equilibrium without the objects in $\tilde{\sigma}$ contradicting the assumption of strong playability. 

The reasoning for these definitions is to restrict ourselves to games in which each pure strategy should be valid to consider under some given meta-game. This definition is to differentiate this game from the case where a pure strategy is not strictly or even weakly dominated by any strategy, yet said pure strategy would never occur in a competitive environment. The three different forms of playability then define three different ways an $RPS$ game can have every strategy see play in a competitive environment. For weakly playable games, each strategy has a meta-game in which it is played with positive probability. In contrast, for playable games, there is a meta-game such that each strategy is played with a positive probability. In strongly playable games, each meta-game must have each strategy played with positive probability. In the two-player $RPS$ context, we will see that these definitions all agree, yet in the sequel, we will see that showing $k$-strong playability is much more complex than showing $k$-weak playability for more than two players, even when restricting to games that lack weakly dominated strategies.

\begin{ex}[$4$-$RPS$]
To see an example of an unplayable $RPS$ game, consider the French variant of Rock-Paper-Scissors, known as \textit{pierre-papier-ciseaux-puits}, rock still beats scissors, which beats paper, which beats rock, but there is a fourth object, well. Well beats rock and scissors and loses to paper. So, well weakly dominates rock as it gives the same payoff if the opponent plays scissors or paper, but better payoffs if the opponent chooses rock or well. This game is not weakly playable, as if there was a Nash equilibrium in which player $p$ gave positive probability $P(R)$ of choosing rock, player $q$ could mirror $p$'s strategy for paper and scissors but increase $p$'s probability of well by $P(R)$. Player $q$ then has a strictly positive expected value; however, as in any Nash equilibrium of a symmetric zero-sum two-player game, either player can always mirror the other to secure zero expected value, the Nash equilibrium of these games must have zero expected value for both players. Therefore, $q$ has strictly improved his position, and so, the position could not have been a Nash equilibrium. Thus, this game is unplayable. 
\end{ex}

Given these definitions of playable, we can now contrast them with the following pair of definitions of balance for a game.
\subsubsection*{Two Forms of Imbalance of a Game}
The first form of imbalance that we consider is based on the fact that in a balanced game, we would expect each object in the game to win as often as they lose. Therefore, the more spread the difference between these values are, the farther the game is from being balanced.

\begin{defi}[Combinatorial Balance] \label{combimb}
 For $n\in \N$, an $n$-$RPS$, $G_1$, is more imbalanced than another $n$-$RPS$, $G_2$, if the difference of outgoing and incoming edges in $G_1$ has greater variance than that of $G_2$.
\end{defi}

The combinatorial aspects of the underlying tournament entirely define this form of balance, while the following definition is based on the Nash equilibria of the underlying game. This second form of imbalance is based on the idea that in a maximally imbalanced game, all players play the same object and continually tie. In contrast, in the maximally balanced game, the number of ties in symmetric equilibria should be minimized. As different Nash equilibria have different numbers of ties, we should choose the symmetric Nash equilibria that minimize expected ties, corresponding to the idea that this game can be played with a \textit{balanced} meta-game. 

\begin{defi}[Distributional Imbalance] \label{distrimb}
     For $n\in \N$, an $n$-$RPS$, $G_1$, is more imbalanced than another $n$-$RPS$, $G_2$, if under their symmetric Nash equilibria which minimizes ties, $n_1,n_2$, the expected number of ties in $n_1$ is higher than the expected number of ties in $n_2$.
\end{defi}

 The intuition behind both of these definitions is that $G_1$ is more imbalanced than $G_2$ if the \textit{good} objects in $G_1$ beat more objects than the \textit{good} objects of $G_2$, and the \textit{bad} objects in $G_1$ lose to more objects than the \textit{bad} objects in $G_2$. Where a \textit{good} object beats more objects than it loses to, and a \textit{bad} object loses to more objects than it beats. In non-playable games, it is easy to make arbitrarily imbalanced games; however, in playable games, all objects can be played, and so, there must be a Nash equilibrium where even the worst object must have a probability of being played.

In section \ref{seinfeldunfairsect}, we expand these definitions to more general games and define two classes of imbalance definitions: one based on the combinatorial balance and the underlying structure of the game, and the other based on the distributional imbalance and the structure of the Nash equilibria of the game. In section \ref{imbalanced}, we provide a game that we show in section \ref{unfairsect} maximizes both imbalance definitions above, as well as the others found in section \ref{seinfeldunfairsect}.

\subsection{Finding Nash Equilibria of $RPS$ Games}

Here, we will show how the playability definitions when restricted to two-player $RPS$ games are equivalent to both players playing totally mixed strategies in each Nash equilibrium, or equivalently, that the underlying tournament is strong \cite{Stockmeyer2022TournamentScoreSequences}. More precisely, we will show the following lemma:

\begin{lem} \label{RPS-is-trivial}
    In a weakly playable $RPS$, $G$, there is a unique Nash equilibrium, which must be symmetric. Thus, if $G$ is weakly playable, it must also be $2$-strongly playable. In particular, this condition is equivalent to the unique Nash equilibrium being totally mixed for both players.
\end{lem}

 \subsubsection*{Expected Payoffs in Nash Equilibria of Symmetric Zero-Sum Games are Zero}
 The payoff matrix of a two-player symmetric zero-sum game with $n$ strategies is an antisymmetric $n\times n$ matrix, such that the $(i,j)$-th entry corresponds to the payoff for the player who played strategy $i$ against a player who plays strategy $j$. 

\begin{defi}
    A probability vector is a vector in the simplex, i.e., a vector with non-negative entries, such that the sum of all entries is $1$. We let $\Delta^{n-1}$ be the set of probability vectors in $\R^n$.
\end{defi}
 
 If we take the payoff matrix for player $q$, $A_q$, and multiply it by the probability vector corresponding to a mixed strategy of player $p$, $\mathbf{v}_p$, the result,  $A_q\mathbf{v}_p$, has entries, $(A_q\mathbf{v}_p)_i$, corresponding to the expected payoffs for the pure strategies of player $q$ \cite{gametheorytextbook}. As the games we consider are symmetric, define $A=A_q=A_p$. In a Nash equilibrium of a zero-sum symmetric game, each of these entries $(A\mathbf{v}_p)_i$ must be non-positive. For if one of these had a positive value, then when keeping player $p$'s strategy constant, player $q$ would have a positive expected value with the pure strategy of playing that object. This contradicts the fact that in a symmetric zero-sum two-player game, the Nash equilibrium must give each player zero net payoff \cite{gametheorytextbook}.

Thus, a pair of probability vectors, $(\mathbf{v}_p,\mathbf{v}_q)$, corresponds to a Nash equilibrium only if both $A\mathbf{v}_p$ and $A\mathbf{v}_q$ have non-positive entries. Let $N(A)$ be the set of probability vectors, $\mathbf{v}_p$, such that $A\mathbf{v}_p$ has only non-positive entries. Then, $N(A)\times N(A)$ must contain the set of all Nash equilibria for the given symmetric game. Moreover, for any two such vectors $\mathbf{v}_p,\mathbf{v}_q$, as $A\mathbf{v}_p$ and $A\mathbf{v}_q$ have only non-positive entries:
\begin{equation}
    E_p(\mathbf{v}_p)=\mathbf{v}_p\cdot(A\mathbf{v}_q) \implies E_p(
    \mathbf{v}_P)\leq 0
\end{equation}
Where $E_p(\mathbf{v}_p)$ is the expected payoff for player $p$ playing the mixed strategy corresponding to $\mathbf{v}_p$.

Similarly $E_q(\mathbf{v}_q)\leq 0$. However, as these games are zero-sum, 
 \begin{equation}
     E(\mathbf{v}_p)=E(\mathbf{v}_q)=0
 \end{equation}
So neither player can increase their expected outcome, and any element  $(\mathbf{v}_p,\mathbf{v}_q) \in N(A)\times N(A)$ is a Nash equilibrium. 

\subsubsection*{Nash Equilibria in Playable Symmetric Zero-Sum Games are Kernels of the Payoff Matrix}
    
If in a Nash equilibrium, of a symmetric game with payoff matrix $A$ player $p$ plays probability vector $\mathbf{v}_p$, such that $(A\mathbf{v}_p)_i$ is negative, then player $q$'s probability of playing this object, $v_{q,i}$ must be $0$. If not, as $\mathbf{v}_q$ is part of a Nash equilibrium, $E_q(\mathbf{v}_q)=0$. As object $i$ provides negative expected value, $v_{q,i}\cdot A(\mathbf{v}_p)_i<0$, there must be an object, $j$, such that, $v_{q,j}\cdot A(\mathbf{v}_p)_j>0$, then playing only $j$ is a strict improvement in $q$'s strategy. Therefore, as long as $p$ is playing $\mathbf{v}_p$, $q$ cannot play any strategy containing $i$ with positive probability in a Nash equilibrium. As each vector $\mathbf{v}$ in $N(A)$ has a Nash equilibrium given by $(\mathbf{v}_p,\mathbf{v})$, $v_i=0$ for all $\mathbf{v} \in N(A)$. Thus, a game with such an object cannot be weakly playable. By the contrapositive, in a weakly playable RPS, each object must have exactly zero expected payoff in every Nash equilibrium, i.e., for all $\mathbf{v}\in N(A)$,
\begin{equation}
    A\mathbf{v}=\mathbf{0}
\end{equation}
We have then shown the following lemma:

\begin{lem}
    The set of Nash equilibria of a weakly playable symmetric zero-sum two-player game must be $(ker(A)\cap \Delta^{n-1}) \times (ker(A)\cap \Delta^{n-1})$, i.e., $N(A)=ker(A)\cap \Delta^{n-1}$
\end{lem}

Lemma \ref{evenrps} in the following subsection demonstrates that any even-dimensional payoff matrix of an $RPS$ game has a trivial kernel and, consequently, no such vectors. Given an $(2n+1)$-$RPS$, in which we have a Nash equilibrium, $(\mathbf{v}_p,\mathbf{v}_q)$, in which $v_{p,i}=0$. Consider a similar $2n$-$RPS$ except that it does not contain object $i$. In this $RPS$, let $\tilde{\mathbf{v}}_p$ be the vector giving the same probabilities to all other objects as $\mathbf{v}_p$. $\tilde{\mathbf{v}}_p$ is still a probability vector, and in the kernel of $\tilde{A}$. Here $\tilde{A}$ corresponds to the even-dimensional payoff matrix of this $2n$-$RPS$, which is similar to the payoff matrix $A$ for the original $(2n+1)$-$RPS$ except that we have deleted the row and column corresponding to object $i$ from $A$. Lemma \ref{evenrps} contradicts the existence of such a vector. Therefore, each Nash equilibrium of a playable $RPS$ must have both strategies be totally mixed, i.e., all such games are $2$-strongly playable.

For any two linearly independent vectors, $\mathbf{v}_1,\mathbf{v}_2$, in the kernel of $A$, there is a linear combination $\mathbf{v}=c(c_1\mathbf{v}_1+c_2\mathbf{v}_2)$ such that $\mathbf{v} \in \Delta^{n-1}$, and some entry $v_i=0$. This would construct a kernel for $\tilde{A}$, which is again a contradiction. Therefore, in playable $RPS$ games, the kernel is exactly 1-dimensional, and there is a unique distribution over the set of pure strategies for players $p$ and $q$ that is a Nash equilibrium, showing lemma \ref{RPS-is-trivial}. 

\subsubsection*{All $2n$-RPS Games are Unplayable}
For an even number of objects, $2n$, it is impossible to make a $2n$-$RPS$ game with $ker(A)\cap \Delta^{2n-1}$ nonempty. Take, for instance, a $4$-$RPS$. This game can never be playable, as one of the following must be the case:
\begin{enumerate}
    \item The game contains a copy of 3-RPS. In which case, the fourth object must weakly dominate or be weakly dominated by some object in the 3-RPS, as we saw in the Rock-Paper-Scissors-Well example above.
    \item Every set of three objects in the game is strictly ordered, which implies that all four objects must be strictly ordered.
\end{enumerate}

For the general case, we will show that the class of payoff matrices for $2n$-$RPS$ games has a trivial kernel. As the Nash equilibrium probability vector must be contained in the kernel and cannot be the zero vector, this is enough to show that all such games are unplayable. To do so, first, we notice that the payoff matrix is always skew-symmetric, with all off-diagonal entries being $\pm 1$. We will use the following stronger lemma to show that all such matrices have a positive determinant and therefore a trivial kernel.

\begin{lem}\label{evenrps}
    All $2n$-dimensional skew-symmetric matrices, in which each off-diagonal entry can be reduced to the form $\frac{a}{b}$ for both $a$ and $b$ odd integers, have a determinant, which is of the form $\frac{a^2}{b^2}$ for $a$ and $b$ odd integers. More particularly, this determinant is non-zero.
\end{lem}

\begin{proof}
    Before beginning the proof, we define an odd rational as a rational number that can be written as $\frac {a}{b}$ for both $a$ and $b$ odd integers; an odd matrix as a matrix containing only odd rationals; and an odd skew-symmetric matrix as a skew-symmetric matrix for which each off-diagonal entry is an odd rational. Similarly, define a rational number that can be written $\frac{a}{b}$ for $a$ even and $b$ odd, an even rational; a matrix containing only even rationals, an even matrix; and a skew-symmetric matrix for which each off-diagonal entry is an even rational, an even skew-symmetric matrix. Note that odd and even rationals add and multiply as expected.

    We will prove this lemma by induction on $n$. When $n=1$, the matrix is of the form:
    \begin{equation}\begin{bmatrix}
        0 & \frac{a}{b} \\
        \frac{-a}{b} & 0 
    \end{bmatrix}.\end{equation}
    This matrix's determinant is $\frac{a^2}{b^2}$, which is the square of an odd rational. Assume that the lemma is true for all $2n$-dimensional odd skew-symmetric matrices. Let $F$ be a $(2n+2)$-dimensional odd skew-symmetric matrix. We can rewrite $F$ in block form as:
    \begin{equation}F=
        \begin{bmatrix}
          A & C \\
          -C^\top & D  
    \end{bmatrix}\end{equation}
    For $A$ a $2\times 2$ matrix, $D$ a $2n\times 2n$ dimensional matrix, and $C$ a $2n\times 2$ matrix. $A$ and $D$ are both odd skew-symmetric matrices, and $C$ is an odd matrix. Let $A$ be given by     
    \begin{equation}A=\begin{bmatrix}
        0 & \frac{a}{b} \\
        \frac{-a}{b} & 0 
    \end{bmatrix}.\end{equation}
    As $A$ is invertible, we can apply the determinant formula for block matrices to $F$. This results in:    
    \begin{equation}\det(F)=\det(A)\det(D-(-C^\top )A^{-1}C)=\frac{a^2}{b^2}\det(D+(C^\top A^{-1}C))\end{equation}
So $\det(F)$ is a square of an odd rational, if $\det(D+(C^\top A^{-1}C))$ is, as the product of squares of an odd rational is a square of an odd rational. By induction, if $D+(C^\top A^{-1}C)$ is a $2n$-dimensional odd skew-symmetric matrix, its determinant is the square of an odd rational. As $A$ is skew-symmetric, so is $A^{-1}$; moreover:
\begin{equation}(D+(C^\top A^{-1}C))^\top =-(D+(C^\top A^{-1}C))\end{equation} 
Thus, as both $D+(C^\top A^{-1}C)$ and $(C^\top A^{-1}C)$ are skew-symmetric, it remains to show that all non-diagonal entries in $D+(C^\top A^{-1}C)$ are odd rationals.

    First, note that,
    \begin{equation}A^{-1}=\begin{bmatrix}
        0 &\frac{-b}{a} \\
        \frac{b}{a} & 0
    \end{bmatrix}\end{equation}
And, for some odd vectors, $\mathbf{V}$, and $\mathbf{W}$, let $C$ be given by:
\begin{equation}C=
\left[
\begin{array}{c}
\begin{tikzpicture}[baseline=(v.base)]
  \node[draw, minimum width=2.5cm, minimum height=0.1cm] (v) { \(\mathbf{V}^\top\) };
\end{tikzpicture}
\\[0.3cm]
\begin{tikzpicture}[baseline=(w.base)]
  \node[draw, minimum width=2.5cm, minimum height=0.1cm] (w) { \(\mathbf{W}^\top\) };
\end{tikzpicture}
\end{array}
\right]
\end{equation}
We can then define the $2n\times 2$ matrix $E$ as: 
    \begin{equation}E=A^{-1}C= \left[
\begin{array}{c}
\begin{tikzpicture}[baseline=(v.base)]
  \node[draw, minimum width=2.5cm, minimum height=0.1cm] (w) { \(\frac{-b}{a}\mathbf{W}^\top\) };
\end{tikzpicture}
\\[0.3cm]
\begin{tikzpicture}[baseline=(w.base)]
  \node[draw, minimum width=2.5cm, minimum height=0.1cm] (v) { \(\frac{b}{a}\mathbf{V}^\top\) };
\end{tikzpicture}
\end{array}
\right]
   \end{equation}
   
Each entry of $E$ is a product of two odd rationals. The first being either $\frac{-b}{a}$, or $\frac{b}{a}$, and the second an entry in either $\mathbf{V}$ or $\mathbf{W}$. Therefore, each entry, as a product of an odd rational and an odd rational, is an odd rational, and $E$ is an odd matrix. Note that the $ith$ column of $C^\top $ is a pair of odd rationals $(f_1,f_2)$ and the $j$th row of $E$ is also a pair of odd rationals $(g_1,g_2)$. The $(i,j)$th entry in $C^\top E$ is the dot product of these vectors, $f_1g_1+f_2g_2$. As a product of odd rationals is odd, this dot product is a sum of two odd rationals, which is an even rational. Therefore, $C^\top E$ is an even skew-symmetric matrix.
So, as the sum of an odd skew-symmetric matrix, $D$, and an even skew-symmetric matrix, $C^\top E$, is an odd skew-symmetric matrix, by induction, $\det(D+(C^\top A^{-1}C))$ is the square of an odd rational. Therefore, $F$ is the product of squares of odd rationals and thus a square of an odd rational, and the lemma is proven.
    \end{proof}

Note that the square root of the determinant for skew-symmetric matrices is called the Pfaffian, and the above shows that under these assumptions, the Pfaffian is always an odd rational.

\section{Alternative Imbalance Definitions} \label{seinfeldunfairsect}
The above definitions of imbalance are far from the only viable definitions and do not easily extend to games with more than two players. The following alternative definitions provide different metrics and partial orders for the imbalance of certain types of games. We can think of these various metrics in the same way we consider different measures of inequality in a wealth distribution of a population. For instance, in a perfectly equal population, we would expect wealth to be uniformly distributed. We can measure our deviation from this uniformity by using metrics such as the GINI coefficient, the coefficient of variation, or Theil's entropy index \cite{Cowell2000_MeasuringInequality}\cite{NinoZarazua2017_GlobalInequality}. We can also see that broadly these imbalances decrease as the number of pure strategies in the games increases, showing that attempting to reduce imbalance overall may cause evolutionary pressure. For instance, a model ecological game with $m$ strategies may evolve into one with more strategies while remaining playable to reduce the imbalance of the game. For this reason, we consider the maximization and minimization of imbalance when restricted to the set of $RPS$ games on a specified number of objects, e.g., the most imbalanced $n$-$RPS$ as opposed to the most imbalanced $RPS$ as a whole. 

Here, we provide some alternative definitions of imbalance that are similar to some of these measures of inequality and extend the above $RPS$ definition to a more general class of symmetric zero-sum games. 

\subsection{Uniform Strategy Imbalance}

Ideally, in a balanced $RPS$ game, there is a Nash equilibrium in which each player chooses each object uniformly at random. In such a game, if all players other than $p$ are playing a uniformly mixed strategy over the $n$ pure strategies, $p$'s expected payoffs for each object should be zero. We can then define:

\begin{defi}[Uniform Expected Payoff]
The uniform expected payoff for player $i$ in a game $G$ is the expected payoff for each pure strategy of player $i$, given that all other players are playing a uniform probability distribution over their set of pure strategies. Therefore, this is only defined if the set of pure strategies for each player is compact. Alternatively, we can define the uniform expected payoff distribution for $i$ by considering these expected values as a distribution giving equal weight to each object's uniform expected payoff.
\end{defi}

For symmetric games, this result is the same for any player. For asymmetric games, we will, in general, use the weighted sum uniform expected payoff to define our imbalance statistic, i.e., the imbalance statistic is based on weighing the statistic based on uniform expected payoff for each player equally. The idea behind this is that the more players for whom the game is imbalanced, the more imbalanced the game. We refer to an imbalance statistic based on the uniform expected payoff as a uniform strategy imbalance. For the sake of clarity, we only show the definitions for symmetric games below.

This statistic applied to an $n$-$RPS$ provides each pure strategy, $i$, with $e_{in,i}$ incoming edges and $e_{out,i}$ outgoing edges a payoff of: 
\begin{equation}
    \frac{1}{n-1} (e_{in,i}-e_{out,i})
\end{equation}  
This directly extends the difference between incoming and outgoing edges in the combinatorial imbalance defined above in definition \ref{combimb}. For this purpose, we consider these forms of imbalance combinatorial rather than game-theoretic. The uniform expected payoff statistic leads to the following several different measures of imbalance:

Firstly, we can define the variance of uniform expected payoffs as a random variable over the set of expected payoffs. As opposed to a balanced game where the variance is $0$, in a maximally imbalanced game, this variance should be maximized. In analogy to the imbalance in wealth distribution, this can be considered similar to the coefficient of variation, which measures the standard deviation of the wealth distribution \cite{NinoZarazua2017_GlobalInequality}. More precisely, we define:

\begin{defi}[Uniform Imbalance of Variance]
For two symmetric games, $G_1$ and $G_2$, $G_1$ is more uniformly imbalanced in variance, referred to as $UI_v$, than $G_2$ if the variance of the uniform expected payoffs in $G_1$ is higher than the variance of the uniform expected payoffs in $G_2$. 
\end{defi}

The partial ordering of imbalance defined by $UI_v$ is equivalent to the combinatorial imbalance definition above for two-player $n$-$RPS$ games. For this reason, $UI_v$ is a reasonable choice of imbalance definition for any game whose sets of strategies are compact. Another form of imbalance based on the uniform expected payoff has to do with maximizing the entropy as a probability distribution over expected payoffs; as in the balanced case, this entropy is minimized. More precisely, we can define:

\begin{defi}[Uniform Imbalance of Entropy]
 For two symmetric games, $G_1$ and $G_2$, $G_1$ is more uniformly imbalanced in entropy, referred to as $UI_e$, than $G_2$ if the entropy of the uniform expected payoffs as a random variable in $G_1$ is higher than the entropy of the uniform expected payoffs in $G_2$. Where the entropy, $H$, of a probability distribution, $P$, over a measurable set $D$, is defined as \cite{Shannon1948MathematicalTheoryCommunication}:
 \begin{equation}
     H(P)=-\int_D P(x)\ln(P(x))dx
 \end{equation}
\end{defi}

A final form of imbalance, based on the uniform expected payoff, involves maximizing a form of the Theil-T index of the resulting function \cite{Theil1967_EconomicsAndInformationTheory}, as opposed to its variance as a probability distribution over expected payoffs. To do so, we must standardize the expected payoffs so that the natural logarithm of an expected payoff is well-defined. For this purpose, we have the definition:  

\begin{defi}[\alpha-Thiel Entropy Imbalance ]
For a symmetric game, $G_1$ with bounded payoffs, define $\tilde{G}$ as a game in which we scale the payoffs of each outcome in $G$ by a positive constant $c_1$ and then add a constant $c_2$ to each payoff, so that the uniform expected payoffs has mean $1$ and infimum value $0<\alpha<1$. Note that this transformation does not change the Nash Equilibria from $G_1$ to $\tilde{G}_1$. 
For two games with bounded payoffs, $G_1$, $G_2$, $G_1$ is more uniformly imbalanced in $\alpha$-Theil $T$ entropy, referred to as $UI_{t_\alpha}$, than $G_2$ if the Thiel-T index of the random distribution over uniform expected payoffs for $\tilde{G_1}$ is higher than the Thiel-T index over $\tilde{G}_2$. Where we define the Thiel-T index of a distribution, $P(x)$ over positive values with mean $1$ as:
\begin{equation}
    T(P)=\int_0^\infty P(x) x\ln(x) dx
\end{equation}
Over discrete random variables defined on a set $D$, with probability distribution $P(x)$ as:
\begin{equation}
    T(P)=\sum_{i\in D} P(x_i)x_i\ln(x_i)
\end{equation}
\end{defi}
Note that over a maximally balanced game, the Theil-T statistic is minimized at $\alpha\ln(\alpha)$ as expected.

\subsubsection*{The Schur-Uniform Class of Imbalances for $RPS$ Games}

There are many other definitions of imbalance that we can define based on uniform expected payoff. In two-player $(2n+1)$-$RPS$ games, the uniform expected payoff given by an object $i$ is a positive proportion of the incoming edges to $i$, $e_{in,i}$, subtracted by the outgoing edges from $i$, $e_{out,i}$. As these sum up to $2n$, we can set the uniform expected payoff as a function of either the sequence of incoming edges, $\mathbf{e}_{in}$, or outgoing edges, $\mathbf{e}_{out}$. We can then consider the following class of imbalance functions:

\begin{defi}[Schur-Uniform Class of Imbalances over $RPS$ games]
    A uniform strategy imbalance, $\tau$, is in the Schur-uniform class, $S_u$, if:

    \begin{enumerate}
        \item $\tau$ is definable over an $n$-$RPS$ game for any $n\in \N$.
        \item  When $\tau$ is applied to $n$-$RPS$, $G_1$, and $G_2$ with sequences of incoming edges $\mathbf{e}_{G_1,in},\mathbf{e}_{G_2,in}$ respectively, if $\mathbf{e}_{G_1,in}$ majorizes $\mathbf{e}_{G_2,in}$ then $G_1$ is greater than or equally imbalanced as $G_2$.
    \end{enumerate}
    If, when $\mathbf{e}_{G_1,in}$ strictly majorizes $\mathbf{e}_{G_2,in}$, $G_1$ is strictly more imbalanced than $G_2$ we say that said imbalance definition is in the strict $S_u$ class.
\end{defi}

Before defining majorization, we must first define for $\mathbf{x},\mathbf{y}\in \R^d$, $x_{i}^{\downarrow }$ as the $i$th largest index of $\mathbf{x}$, and similarly define $y_{i}^{\downarrow }$. To simplify notation, we also define the reverse ordered sequence of $\mathbf{x}$, $\mathbf{x}^\downarrow$, as the sequence $(x_{1}^{\downarrow },x_{2}^{\downarrow }\dots)$, and the ordered sequence of $\mathbf{x}$, $\mathbf{x}^\uparrow$, as $(x_{1}^{\uparrow },x_{2}^{\uparrow }\dots)$, for $x_{i}^{\uparrow }$ the $i$th smallest index in $\mathbf{x}$.
Then $\mathbf{x}$ majorizes $\mathbf{y}$ if both of the following hold \cite{MarshallOlkinArnold2011Majorization2e}:
        \begin{equation}
            \sum _{i=1}^{d}x_{i}= \sum _{i=1}^{d}y_{i}
        \end{equation}
And for all natural numbers $k<d$,
        \begin{equation}
             \sum _{i=1}^{k}x_{i}^{\downarrow }\geq \sum _{i=1}^{k}y_{i}^{\downarrow }
        \end{equation}

While $UI_v$ and $UI_{t_\alpha}$ act as different partial orders on games' level of imbalance, if two games are in the $S_u$ class, then on the most imbalanced form of $RPS$ games, they agree.  We can then show the following:

\begin{lem}
    $UI_v$ is in the strict $S_u$ class, and $UI_{t_\alpha}$ is in the $S_u$ class.
\end{lem}

\begin{proof}
    As we saw above, uniform expected payoff in $(2n+1)$-$RPS$ games provides each object a payoff of $\frac{1}{2n}(e_{in,i} -e_{out,i})$. As the coefficient $\frac{1}{2n}$ is the same over all $(2n+1)$-$RPS$ games, we can ignore this coefficient and consider the uniform imbalance statistics defined on the variables $e_{in,i}$ with expected payoff for object $i$ equal to $(e_{in,i} -e_{out,i})$, which as $e_{out,i}=2n-e_{in,i}$ is equivalent to  $(2e_{in,i}-2n)$.

    For $UI_v$, this corresponds to maximizing the variance in $(2e_{in,i}-2n)$ over the objects $i$ with equal probability. As the average of this statistic is $0$ this is equivalent to maximizing $\sum_{i=1}^{2n+1}(2e_{in,i}-2n)^2$. As the mixed second-order partial derivative of this function with respect to $e_{in,i}$ and $e_{in,j}$ is $0$, and the non-mixed second-order partial derivative with respect to $e_{in,i}$ is $2$, this function is convex. As it is also symmetric over each variable of $e_{in,i}$, it is Schur-convex and therefore in $S_u$ \cite{MarshallOlkinArnold2011Majorization2e}.
    
    To see it is in the strict $S_u$ class, note that we can take the partial difference of this with respect to increasing $e_{in,i}$ by $1$ and decreasing $e_{in,j}$ by $1$, as the sum of all incoming edges must be constant:
    \begin{equation}
    \begin{split}
       (2(e_{in,i}+1)-2n)^2-(2e_{in,i}-2n)^2+(
        (2(e_{in,j}-1)-2n)^2-(2e_{in,j}-2n)^2=\\ 
        8\left(e_{in,i}-n+\frac{1}{2}\right)-8\left(e_{in,j}-2n+\frac{1}{2}\right)=\\
        8(e_{in,i}-e_{in,j}+1)
    \end{split}
    \end{equation}
    Thus increasing any $e_i$ that is larger than or equal to some $e_j$ by $1$, and correspondingly decreasing $e_j$ by $1$, increases the $UI_v$ statistic by $8\left(e_{in,i}-e_{in,j}+1\right)$. As the larger these $e_{in,i}$ are, the larger this increase is, and as it is symmetric for each $e_{in,i}$, we can see that a sequence of $e_{in,i}$ that strictly majorizes another will have a strictly larger $UI_v$ statistic.

    For $UI_{t_\alpha}$, we can see that for some $c_1,c_2$ this is equivalent to maximizing:
    \begin{equation}
        \sum_{i=1}^{2n+1}c_1( 2e_{in,i}-2n+c_2)\ln (c_1( 2e_{in,i}-2n+c_2))
    \end{equation}
    We can see that this is symmetric in $e_{in,i}$ and that the mixed second-order partial derivatives of this function with respect to $e_{in,i}$ and $e_{in,j}$ are $0$. The non-mixed second-order partial derivative with respect to $e_{in,i}$ is:
    \begin{equation}
    \frac{4c_1^2}{c_1( 2e_{in,i}-2n+c_2)}
    \end{equation}
    As both the numerator and denominator of this function are always positive by construction, this function is convex. Therefore, this function is Schur-convex and thus maximized by a majorized sequence of variables \cite{MarshallOlkinArnold2011Majorization2e}. Therefore, this imbalance is in $S_u$. 
    \end{proof}
    
In section \ref{unfairsect}, we show that the given imbalanced game's sequence of incoming edges strictly majorizes each other's over playable $RPS$ games, showing that it maximizes $S_u$ imbalance definitions. We will further show that this $RPS$ game also maximizes $UI_e$ among playable games.

\subsection{Nash Probability Imbalance}

Ideally, in a balanced $RPS$ game, there is a Nash equilibrium in which each player chooses each object uniformly at random. Therefore, a measure of deviation of a Nash equilibrium from the uniform distribution defines a form of imbalance. We call these forms of imbalance Nash probability imbalances. As mentioned above, when selecting the Nash equilibrium to define the imbalance of the game, we choose the one that minimizes the imbalance. We call these the worst-case distributions. The most classic consideration that we can choose is to minimize the entropy, giving the definition:

\begin{defi}[Minimal Nash Equilibria Entropy Imbalance]
     Given two $m$-player games, $G_1$, and $G_2$, $G_1$ is more Nash equilibria entropy imbalanced, referred to as $N_e$, than $G_2$ if in the Nash equilibrium, which maximizes the sum of the entropies for all players of $G_1$, and $G_2$, $n_1$, $n_2$ respectively, $H(n_1)<H(n_2)$. Where $H(n_i)$ is the sum of the entropy of each distribution of each player in the set of distributions $n_i$. For instance, if $P_j(x)$ is the distribution over pure strategies for player $p_j$ in $n_i$, then given $m$ players, $p_1 \dots p_m$, each with their own strategy set $D_j$: 
     
     \begin{equation}
         H(n_i)=-\sum_{j=1}^m \int_{D_j} P_j(x)\ln(P_j(x))dx
     \end{equation} 
\end{defi}
Note that over a maximally balanced game containing a uniform distribution Nash equilibrium for each player, this entropy statistic is maximized.

Another definition of imbalance that applies to symmetric discrete games comes from the fact that in an imbalanced game, we would expect more of the players to choose the same pure strategy. This leads us to the extension of distributional imbalance mentioned above:

\begin{defi}[Maximal Nash Equilibria Ties Imbalance]
     Given two $m$-player symmetric discrete games, $G_1$, and $G_2$, $G_1$ is more Nash maximal ties imbalanced, $N_t$, than $G_2$ if in the symmetric Nash equilibrium, which minimizes ties for $G_1$, and $G_2$, $n_1$, $n_2$ respectively, $T(n_1)>T(n_2)$, where $T(n_i)$ is the expected number of ties when players play the set of strategies $n_i$ over $d$ pure strategies. For instance, if $v_{o}$ is the probability of a player playing $o$ in the Nash equilibrium $n_i$ then, 
     \begin{equation}
         T(n_i)=\sum_{o=1}^d {v_o}^m
     \end{equation}   
\end{defi}     
Note that all symmetric games have a symmetric Nash equilibrium \cite{nashnperson}, and that in a uniform distribution for each player, the number of ties is minimized over all symmetric distributions.

\subsubsection*{The Schur-Distributional Class of Imbalances for $RPS$ Games}

There are many other definitions of imbalance that we can define based on Nash equilibrium probability distributions. Just as for the majorization uniform imbalance class, we can also consider the following class of imbalance functions:

\begin{defi}[Schur-Distributional Class of Imbalances over $RPS$ games]
    A Nash probability imbalance, $\tau$, is in the Schur-distributional class, $S_N$, if:

    \begin{enumerate}
        \item $\tau$ is definable over a $n$-$RPS$ game for any $n\in \N$
        \item  When $\tau$ is applied to  $n$-$RPS$s, $G_1$, and $G_2$ with worst-case Nash equilibria probabilities for player $j$ playing object $o$ of $n^j_{G_1,o}$ and $n^j_{G_2,o}$ respectively, if the sequence of all such player-object probabilities for $G_1$ majorizes that sequence for $G_2$ then $G_1$ is greater than or equally imbalanced as $G_2$.
    \end{enumerate}
    If, when a sequence of probabilities for $G_1$ strictly majorizes $G_2$'s sequence of probabilities implies $G_1$ is more imbalanced than $G_2$, we say that said imbalance definition is in the strict $S_N$ class.
\end{defi} 

While $N_e$ and $N_t$ are different forms of imbalance that may disagree on certain games, we will show that these two forms of imbalance are both in $S_N$ and therefore agree on the most imbalanced $RPS$ game, more precisely:

\begin{lem}
    $N_t$ is in the strict $S_N$ class, and $N_e$ is in the $S_N$ class.
\end{lem}

\begin{proof}
    To show that this is the case for $N_t$, note that we are maximizing:     
    \begin{equation}
    \sum_i^{2n+1}(P(o_i)^2)
    \end{equation}
   The partial derivatives of this with respect to any variable $P(o_i)$ is $2P(o_i)$. Therefore, we can see that as all $P(o_i)$ are non-zero in this Nash equilibrium, increasing the largest $P(o_i)$ at the cost of decreasing a smaller $P(o_j)$ increases the proportion of ties.

   To show that this is the case for $N_e$, note that we are minimizing:
    \begin{equation}
    -\sum_i^{2n+1}(P(o_i)\ln(P(o_i))
    \end{equation}
    The mixed second-order partial derivative of this function with respect to $P(o_i)$ and $P(o_j)$ is $0$, and the non-mixed second-order partial derivatives with respect to $P(o_i)$ are:
    \begin{equation}
        -\frac{1}{P(o_i)}
    \end{equation}
    As this is strictly negative, when well-defined, this function is concave. Therefore, as this function is concave and symmetric with respect to $P(o_i)$, it is Schur-concave and thus $N_e$ is in the $S_N$ imbalance class \cite{MarshallOlkinArnold2011Majorization2e}. 
\end{proof}

In section \ref{unfairsect}, we show that our given imbalanced $(2n+1)$-$RPS$'s sequence of Nash equilibrium object probabilities majorizes all other Nash equilibrium object probabilities over all playable $RPS$ games.

\subsection{Uniform Imbalance on an $RPS$ with Countably Infinite Strategies} \label{uncountablesect}

For a non-compact set of strategies, we can still define the Nash probability imbalances, but the uniform strategy imbalances cannot be well-defined. To counteract this and define a combinatorial form of imbalance on an $\N$-$RPS$, we provide the following extension of weak majorization for infinite sequences over the negatively extended real numbers, $\bar{\R}=\R\cup\{-\infty\}$. We then show that over $RPS$ games, this directly extends the finite $UI_v$/combinatorial imbalance definition.

\begin{defi}[Weak Majorization over Infinite Sequences of Negatively Extended Real Numbers]
    Given an infinite sequence of negatively extended real numbers, $\bar{\R}=\R\cup \{ -\infty\}$, $\mathbf{x}$, as before define $x_{i}^{\downarrow }$ as the $i$th largest index of $\mathbf{x}$, and similarly define $y_{i}^{\downarrow }$. Then $\mathbf{x}$ weakly majorizes $\mathbf{y}$ if:
    \begin{enumerate}
        \item The number of indices of  $\mathbf{x}$ which are -$\infty$ is the less than or equal to that of $\mathbf{y}$
        \item If the number of indices of  $\mathbf{x}$ which are -$\infty$ are the same as that of $\mathbf{y}$, then if for all $k\in \N$:
        \begin{equation}
            {\displaystyle \sum _{i=1}^{k}x_{i}^{\downarrow }\ge \sum _{i=1}^{k}y_{i}^{\downarrow }}
        \end{equation}
    \end{enumerate}
    We say $\mathbf{x}$ weakly majorizes $\mathbf{y}$ after $k_0\in \N$ if the number of indices of $\mathbf{x}$ which are -$\infty$ are the same as that of $\mathbf{y}$, and if for all $k>k_0$:
\begin{equation} 
    {\displaystyle \sum _{i=1}^{k}x_{i}^{\downarrow }\geq \sum _{i=1}^{k}y_{i}^{\downarrow }}
\end{equation}  
\end{defi}

We further define the minimal degree sequence of a tournament as:
\begin{defi}[Minimal Degree Sequence of a Tournament]
    Given a labeled tournament, $G$, for each vertex $i$, let $e_{min,i}=min(e_{in,i},e_{out,i})$. Then the sequence of $\mathbf{e}_{G,min}=(e_1,e_2\dots)$ is referred to as the minimal degree sequence of $G$.
\end{defi}
We can then define the uniform imbalance on an $\N$-$RPS$ as: 

\begin{defi}[Uniform Imbalance on an $\N$-$RPS$]
    For two $\N$-$RPS$s, $G_1$, and $G_2$, with minimal degree sequences, $\mathbf{e}_{G_1,min},\mathbf{e}_{G_2,min}$ respectively, if $-\mathbf{e}_{G_1,min}$ majorizes $-\mathbf{e}_{G_2,min}$ then $G_1$ is equally or more imbalanced than $G_2$. Furthermore, we say that $G_1$ is equally or more imbalanced then $G_2$ in the limit if there is some $k_0\in \N$ such that $-\mathbf{e}_{G_1,min}$ majorizes $-\mathbf{e}_{G_2,min}$ after $k_0$. Strict uniform imbalance is defined by if when the majorization is strict, then $G_1$ is strictly more imbalanced than $G_2$. Similarly, we define strict uniform imbalance in the limit. 
\end{defi}

While complicated, this definition is based upon the fact that in a balanced $\N$-$RPS$, each object, $i$, has $e_{in,i}=e_{out,i}=+\infty$. Therefore, in an imbalanced $RPS$, each object should have one of $e_{in,i},e_{out,i}$ small, and the other infinite. This form of weak majorization is a method for finding the $RPS$ with the largest number of small $e_{min,i}$ indices. More precisely we can show how the above definition extends $UI_v$ to $\N$-$RPS$ games:

\begin{proof}
    The most imbalanced $n$-$RPS$ maximizes the variance of the difference between outgoing and incoming edges. For sequences with only finite terms, the average difference between outgoing edges and incoming edges is zero, as every edge is incoming in one direction and outgoing in the other. Therefore, for objects $1\dots n$ the variance is:
    
    \begin{equation}
    \label{N-unfairness}
        \frac{1}{n}\sum_{i=1}^{n} (e_{in,i}-e_{out,i})^2 
    \end{equation}
    As $e_{in,i}+e_{out,i}=n-1$,
    \begin{equation}
        (e_{in,i}-e_{out,i})^2=((n-1)-2e_{in,i})^2=((n-1)-2e_{out,i})^2=((n-1)-2e_{min,i})^2
    \end{equation}
    As we are comparing sums, we can ignore the coefficient $\frac{1}{n}$. Expanding out the sum:    
    \begin{equation} \sum_{i=1}^{n} ((n-1)-2e_{min,i})^2=n\cdot (n-1)^2-4(n-1)\cdot \sum^n_{i=1}(e_{min,i}) + 4\sum^n_{i=1}(e_{min,i})^2\end{equation}    
    We can ignore the $n\cdot (n-1)^2$ contribution, as it is the same contribution in all games with $n$ objects. Therefore, on the minimal sequence $e_{min,i}$, we are maximizing: 
\begin{equation}\lim_{n\rightarrow \infty}-4(n-1)\cdot \sum^n_{i=1}(e_{min,i}) + 4\sum^n_{i=1}(e_{min,i})^2\end{equation} 
If we take the partial derivatives of this sum with respect to $e_{min,i}$, we obtain $-(n-1)+2e_{min,i}$. Fixing $i$ as $n$ approaches infinity, we have $e_{min,i}\leq\frac{n-1}{2}$, and all partial derivatives are non-positive, while the non-mixed second-order partial derivatives are positive. Therefore, the smaller these $e_{min,i}$ are, the more imbalanced the game. We can also show that $\sum_{i=1}^n-(n-1)(e_{min,i})$ dominates over $\sum^n_{i=1}(e_{min,i})^2$ in the limit. As, for any $\epsilon>0$, for every finite $e_{min,i}$ there is some $m$ such that, for all $n>m$, 
    \begin{equation}
        -(n-1)(e_{min,i})(1+\epsilon )<-(n-1)(e_{min,i})+(e_{min,i})^2<-(n-1)(e_{min,i})(1-\epsilon )
    \end{equation} 
So maximizing the sum \eqref{N-unfairness} is equivalent to maximizing $-\sum^n_{i=1} e_{min,i}$, which is equivalent to minimizing $\sum^n_{i=1} e_{min,i}$. For two games with the same amount of non-infinite $e_{min,i}$, this motivates that the game with the smallest sum of the given sequence of $e_{min,i}$ is more imbalanced. 

To extend this to infinite sequences of terms, if a game has more infinite terms than another or fewer finite terms, it should be considered more balanced, as the objects corresponding to these infinite $e_i$ are maximally balanced, and the objects corresponding to the finite $e_i$ are at least somewhat imbalanced. In the case that the number of infinite terms is the same, and there are $d\in \N$ finite terms, we compare $\sum^k_{i=1} e_{min,i}^\uparrow$ as above. On the other hand, if there are infinite finite terms, we have: 
\begin{equation}
    \sum^\infty_{i=1} e_{min,i} =\infty
\end{equation}
Therefore, to compare games $G_1$ and $G_2$, we can consider the difference of the contribution of the least balanced element of $G_1$, in equation \eqref{N-unfairness}, with the least balanced element of $G_2$. So, we can construct the ordered minimal degree sequences for both games $\mathbf{e}_{G_1}^\uparrow$ and $\mathbf{e}_{G_2}^\uparrow$, and then take the difference of each term in the sum, i.e.,  
    \begin{equation}
        \sum^\infty_{i=0}(e_{G_1,i}^\uparrow-e_{G_2,i}^\uparrow)
    \end{equation}
    We can thus conclude that $G_1$ is more imbalanced than $G_2$, if for all $k\in \N$ the partial sums from $0$ to $k$ are negative. On the other hand, if there is some $k_0 \in \N$ such that for all $k>k_0$, the partial sum from $0$ to $k$ is negative, we say that this game is imbalanced in the limit. This limit case implies that for all $k>k_0$, the $k$ least balanced objects in $G_1$ are less balanced (contribute less to equation \eqref{N-unfairness}) than the $k$ least balanced objects in $G_2$. If there is some $k_0$ such that all greater partial sums are positive, then the reverse is true. However, if there is no such $k_0$, then the games should be considered incomparable.
\end{proof}

 The limit imbalance consideration allows us to compare more games while only ignoring a finite number of imbalanced elements. 

\section{Imbalanced RPS Constructions}\label{imbalanced}
It is easy to see that having a single object that always wins would be unplayable. In that vein, consider a $(2n+1)$-$RPS$ containing an object, $r_1$, that beats all objects other than $p_1$. Note that $p_1$ must lose to every other object, as if it beats object $s'$, then $s'$ is weakly dominated by $r_1$. As a weakly playable $(2n+1)$-$RPS$ is strongly playable, this cannot occur. 

To maximize the imbalance of the remaining objects' interactions, label our next object that is imbalancedly winning as $r_2$. Just like for $r_1$, if $r_2$ beat all remaining objects, it would weakly dominate them, and that would make the game unplayable. Therefore, there needs to be an object from the remaining choices that beats $r_2$, which we will denote as $p_2$. The same property as for $p_1$, namely being forced to lose to the remaining objects, is present for $p_2$. We can continue this process $n$ times to define the interactions between $2n$ objects. However, as there must be an odd number of objects, we include a final object, $s$, whose interactions with all other objects are well-defined by the construction, e.g., loses to all $r_i$ and beats all $p_i$. We can then define our potentially very imbalanced game:

\begin{defi}[Imbalanced $(2n+1)$-$RPS$]
    This $(2n+1)$-$RPS$ game is defined on a set of vertices $r_1\dots r_n,p_1,\dots p_n,s$. Defined recursively, $r_i$ beats all elements $r_j,p_k,s$ for $i>j,i>k$, $p_i$ loses to all elements $r_j,p_k,s$ for $i>j,i>k$, and $p_i$ beats $r_i$. This implies that $s$ loses to all $r_i$ and beats all $p_i$, and therefore can be thought of as equivalent to both $p_{n+1}$ and $r_{n+1}$. 
\end{defi}

    This construction easily generalizes as an $\N$-$RPS$. 
    
    \begin{defi}[Imbalanced $ \N$-$RPS$]
        The imbalanced $\N$-$RPS$ is defined on a set of vertices indexed by $2$ copies of the natural numbers. The $i$th element in the first copy of the natural numbers is labeled $r_i$, and the $j$th element in the second copy is labeled $p_j$. Defined recursively, $r_i$ beats all elements $r_j,p_k$ for $j>i,k>i$, and $p_i$ loses to all elements $r_j,p_k$ for $j>i,k>i$. Lastly, $p_i$ beats $r_i$. 
    \end{defi}

\begin{lem}
    Both the imbalanced $(2n+1)$-$RPS$ and imbalanced $\N$-$RPS$ constructions are playable. 
\end{lem}

\begin{proof}
To see that each of these games is playable, we need to find a probability vector in the kernel of the payoff matrix corresponding to a positive probability of choosing each object. As the only Nash equilibrium is symmetric, let $P(r_i)$ be the probability of player $q$ playing $r_i$ without loss of generality. As the expected payoff of $r_1$ must be $0$, playing $r_1$ must win as often as it loses. Therefore, the proportion of losses from playing $r_1$ given that the result was not a tie is $\frac{1}{2}$. Let $A$ be the statement ``the player playing  $r_1$ loses", and $B$ the statement ``the player playing $r_1$ did not tie". Then, by Bayes' theorem:
\begin{equation}P(A| B)= \frac{P(B|A)\cdot P(A)}{P(B)}= \frac{1\cdot P(p_1)}{1-P(r_1)}=\frac{1}{2}\end{equation}
Similarly the probability that $p_1$ wins given that it did not tie is also $\frac{1}{2}$ resulting in:
 \begin{equation}\frac{P(r_1)}{1-P(p_1)}=\frac{1}{2}\end{equation}
 Solving for $P(p_1)$ using the first equation gives that: 
 \begin{equation}P(p_1)= \frac{1}{2}-\frac{P(r_1)}{2}\end{equation}
  Inputting this into the second equation results in: 
 \begin{equation}P(r_1)=\frac{1}{4}(1+P(r_1))\end{equation}
Which solves to: 
 \begin{equation}P(r_1)=\frac{1}{3},\, \,P(p_1)=\frac{1}{3}\end{equation}
 
Note that the sum of the remaining probabilities must be $\frac{1}{3}$, and that any other object $l$ must beat $p_1$ one-third of the time and lose to $r_1$ one-third of the time. Therefore, the expected payoff contribution from the opponent playing $r_1$ or $p_1$ in this Nash equilibrium on the pure strategy of playing object $l$ is zero. Assume, by way of induction, that for some $i<n$, for all $j<i$, that:
\begin{equation}
    P(p_j)=P(r_j)= \frac{1}{3^j}
\end{equation}
The base case is above. Note that if the induction holds, then the sum of the probability of playing the $p$-type and $r$-type objects is the same, showing that the expected payoff for the pure strategy of playing $s$ is zero. The remaining proportion that must be allocated to $s$ so that the sum of the probabilities is $1$ is $\frac{1}{3^{n}}$. Thus, if the induction holds, we have shown that this imbalanced game is playable.
 
By induction, the contribution to the expected payoff of the pure strategy of playing $r_{i}$ and the pure strategy of playing $p_{i}$ from the probability the opponent plays any $r_j,p_j$ for all $j<i$ is zero. The proportion of probability that is left to allocate is $\frac{1}{3^{i}}=1-\sum_{j=1}^{i-1}2\cdot\frac{1}{3^{j}}$. Let $S_k$ be the statement that the opponent plays any object $r_l,p_l$ for $l\geq k$. As the expected payoff of $r_{i}$ given that $S_i$ is false is zero, the expected payoff of $r_{i}$ given that $S_{i}$ is true must also be zero. So, the probability that the player who plays $r_{i}$ loses given $S_{i}$ is true and that they did not tie is $\frac{1}{2}$. As before, let $A_i$ be the statement that the player who plays $r_i$ loses, and $B_i$ be the statement that the player who plays $r_i$ did not tie.
 
 By iterated Bayes' theorem, this gives us that:

 \begin{equation}P((A_i|B_i\cup S_i)= P((A_i|B_i)| S_i)= \frac{P((B_i|A_i)|S_i))P(A_1|S_i)}{P(B_i|S_i)}=\frac{P(p_i|S_i)}{1-P(r_i|S_i)}=\frac{1}{2}\end{equation}
Similarly, we can get a corresponding equation coming from the expected wins of $p_i$ given similar statements: 
\begin{equation}\frac{P(r_i|S_i)}{1-P(p_i|S_i)}=\frac{1}{2}\end{equation}
As before, solving these yields:
\begin{equation}
    P(p_i|S_i)=P(r_i|S_i)=\frac{1}{3}
\end{equation}
By the inductive assumption:
\begin{equation}
P(S_i)=\frac{1}{3^{i}}=1-\sum_{j=0}^{i-1}2\cdot \frac{1}{3^{i+1}}
\end{equation}
Therefore, by Bayes' theorem:
\begin{equation}P(p_i|S_i)=\frac{P(S_i|p_i)\cdot P(p_i)}{P(S_i)}\end{equation}
As $P(S_i|p_i)=1$ this solves to:
\begin{equation}P(p_i)=P(p_i|S_i)\cdot P(S_i)=\frac{1}{3}\cdot\frac{1}{3^i}=\frac{1}{3^{i+1}}\end{equation}
And similarly:
\begin{equation}P(r_i)=\frac{1}{3^{i+1}}\end{equation}
By induction, the probability of choosing $r_i$ is the same as the probability of choosing $p_i$, which is $\frac{1}{3^{i}}$, which finishes the proof. As the induction holds for any integer $n$, the same argument, ignoring the $s$ object and letting $n$ go to infinity, works for the $\N$-$RPS$ case. 

Note that we can also find this equilibrium by solving for the unique probability vector in the kernel of the payoff matrix of this game. The above approach is equivalent, yet more directly extends to how we will prove Nash distribution majorization below.
\end{proof}

\subsection{This Construction is the Least Balanced}\label{unfairsect}

We will prove the following:
\begin{thm} \label{unfairthm}
    The imbalanced $(2n+1)$-$RPS$ constructed above maximizes each imbalance statistic in the classes $S_N$, $S_u$, as well as $UI_e$ imbalance over all playable games on $2n+1$ objects. The imbalanced $2\N$-$RPS$ similarly maximizes these statistics. Moreover, it is the unique maximum over combinatorial/$UI_v$ and distributional/$N_t$ imbalance as these are in the strict $S_u$ and $S_N$ classes, respectively.
\end{thm}

This theorem demonstrates that the convention that different forms of imbalance, e.g., over wealth, income, etc., should agree over the least balanced examples, by showing the same is true for two separate types of ways to identify the imbalance in games. Note that these definitions were each constructed so as to agree over the most balanced $RPS$ game, and so the convention that different forms of imbalance agree over the most balanced case is confirmed as well.

By the definition of the $S_u$, $S_N$ classes and their strict versions, we must show that the sequence of incoming edges, $\mathbf{e}_{in}$ majorizes all other possible sequences over playable games, and that the sequence of Nash equilibrium probabilities $\mathbf{v}_N$ majorizes all other probabilities over all playable games. Finally, we will use a separate argument to show that the above games maximize $UI_e$. 

\subsubsection*{Maximal Majorization of $e_{in,i}$ over Playable Games, and Maximization of Uniform Imbalance on $\N$-$RPS$}

To begin, we will show the following known lemma\cite{MarshallOlkinArnold2011Majorization2e}:
\begin{lem} \label{minismaj}
    If for all $k\leq n,k'\leq n$, a $(2n+1)$-tournament strictly minimizes both:
    \begin{equation}\label{sums-to-minimize}
        \sum_{i=1}^ke_{in,i}^\uparrow,  \sum_{i=1}^ke_{out,i}^\uparrow,
    \end{equation}
then the sequence of incoming edges $\mathbf{e}_{in}$ strictly majorizes all other sequences of incoming edges $\mathbf{f}_{in}$.
\end{lem}

This lemma will allow us to use the same method of proof for the $\N$-$RPS$ case as the $(2n+1)$-$RPS$ case, despite the different definitions of majorization.

\begin{proof}
    Firstly, we can confirm that the sum of incoming edges is the same over all $(2n+1)$-vertex tournaments. So all that needs to be shown is that for all $m\leq 2n$,
    \begin{equation}\label{maj-equal-minei}
        \sum_{i=1}^m e_{in,i}^\downarrow >\sum_{i=1}^m  f_{in,i}^\downarrow 
    \end{equation}
    As for all $m\leq n$,
    \begin{equation}
        \sum_{i=1}^m e_{out,i}^\uparrow <\sum_{i=1}^m  f_{out,i}^\uparrow 
    \end{equation}
 Then as $e_{in,i}=2n-e_{out,i}$ we have:
     \begin{equation}
        \sum_{i=1}^m e_{in,i}^\downarrow= \sum_{i=1}^m (2n-e_{out,i}^\uparrow) >\sum_{i=1}^m (2n-f_{out,i}^\uparrow)=\sum_{i=1}^m f_{in,i}^\downarrow
    \end{equation}
    This proves the equation \ref{maj-equal-minei} for $m\leq n$. 

    As we have for all positive $m\leq n$:
    \begin{equation}
        \sum_{i=1}^m e_{in,i}^\uparrow <\sum_{i=1}^m  f_{in,i}^\uparrow 
    \end{equation}
    And:
    \begin{equation}
        \sum_{i=1}^{2n+1} e_{out,i}^\uparrow =\sum_{i=1}^{2n+1}  f_{out,i}^\uparrow 
    \end{equation}
    Then we have that:
    \begin{equation}
        \sum_{i=1}^{2n+1-m} e_{out,i}^\downarrow=\sum_{i=1}^{2n+1} e_{out,i}^\uparrow - \sum_{i=1}^m e_{in,i}^\uparrow >\sum_{i=1}^{2n+1}  f_{out,i}^\uparrow -\sum_{i=1}^m  f_{in,i}^\uparrow=\sum_{i=1}^{2n+1-m} f_{out,i}^\downarrow
    \end{equation}
    This completes the proof by showing equation \ref{maj-equal-minei} for $m$ such that $n\leq m\leq 2n$.
\end{proof}

Therefore, if we can show that the given imbalanced $(2n+1)$-$RPS$ minimizes sums \ref{sums-to-minimize}, and that the minimal sequence, $\mathbf{e}_{min}$ of the $\N$-$RPS$ has no infinite elements and minimizes the sum over all $k\in \N$:
\begin{equation}
    \sum_{i=1}^k e_{min,i}^\uparrow
\end{equation}

We will begin by proving the following corollary to Landau's theorem of strong tournament sequences \cite{Landau1953ScoreStructure} \cite{Moon1963ExtensionLandau}:
\begin{cor}\label{mincor}
 In the finite case, for all $k\leq n$, we have that
 \begin{equation}
     \sum_{i=1}^k e_{out,i}^\uparrow \leq \frac{k(k+1)}{2},\sum_{i=1}^{n+1} e_{out,i}^\uparrow \leq \frac{(n+1)n}{2}+n
 \end{equation}
 In the infinite case, for all $k\in \N$,
 \begin{equation}
     \sum_{i=1}^k e_{out,i}^\uparrow \leq \frac{k(k+1)}{2}
 \end{equation}
For both games, the same is true for the sum of the smallest $k$ indices of $\mathbf{e}_{in}$.
 \end{cor}

We will provide this proof using the following lemma:
\begin{lem}\label{kmin}
    Let $G$ be a $(2n+1)$-$RPS$, $k\leq n+1$. Let the first $k$ objects with the smallest $e_{out,i}$ be the $k$-minimizing objects. $G$ is playable only if each $k$-minimizing object has at least one object that is not $k$-minimizing beating it, or the set of objects beating the $k$-minimizing objects consists of the rest of the objects in the game. The same statement holds for any $\N$-$RPS$, and any $k\in \N$. Moreover, a similar restriction is present for the $k$-maximizing objects.
\end{lem}

To motivate this lemma, consider that the smallest that $e_{out,1}$ can be in any playable game is $1$; as, if an object had $0$ outgoing edges, it would weakly dominate all other objects in the game. For $n>1$, if two objects $a,b$ had $e_{out,b}=e_{out,b}=1$ and, without loss of generality, assume that object $a$ beats object $b$, then there is some unique object, $o$, which beats $a$ and loses to $b$. For any other object, $i$, objects $a$ and $b$ both beat $i$. As $b$ only loses to $a$, $b$ weakly dominates $i$. Thus, for this to be playable, $b$ should have an object that is not $a$ beat it.

\begin{proof}
 
 Consider, by way of contradiction, that for a $(2n+1)$-$RPS$, there is some $k\leq n+1$, or for a $\N$-$RPS$, some $k\in \N$, such that this lemma does not hold. Then there is some $k$-minimizing object, $b$, that beats all non-$k$-minimizing objects, and the objects that directly beat the $k$-minimizing objects do not constitute the entire rest of the game. 

\begin{figure}[ht]
    \centering
    \includegraphics[width=.5\linewidth]{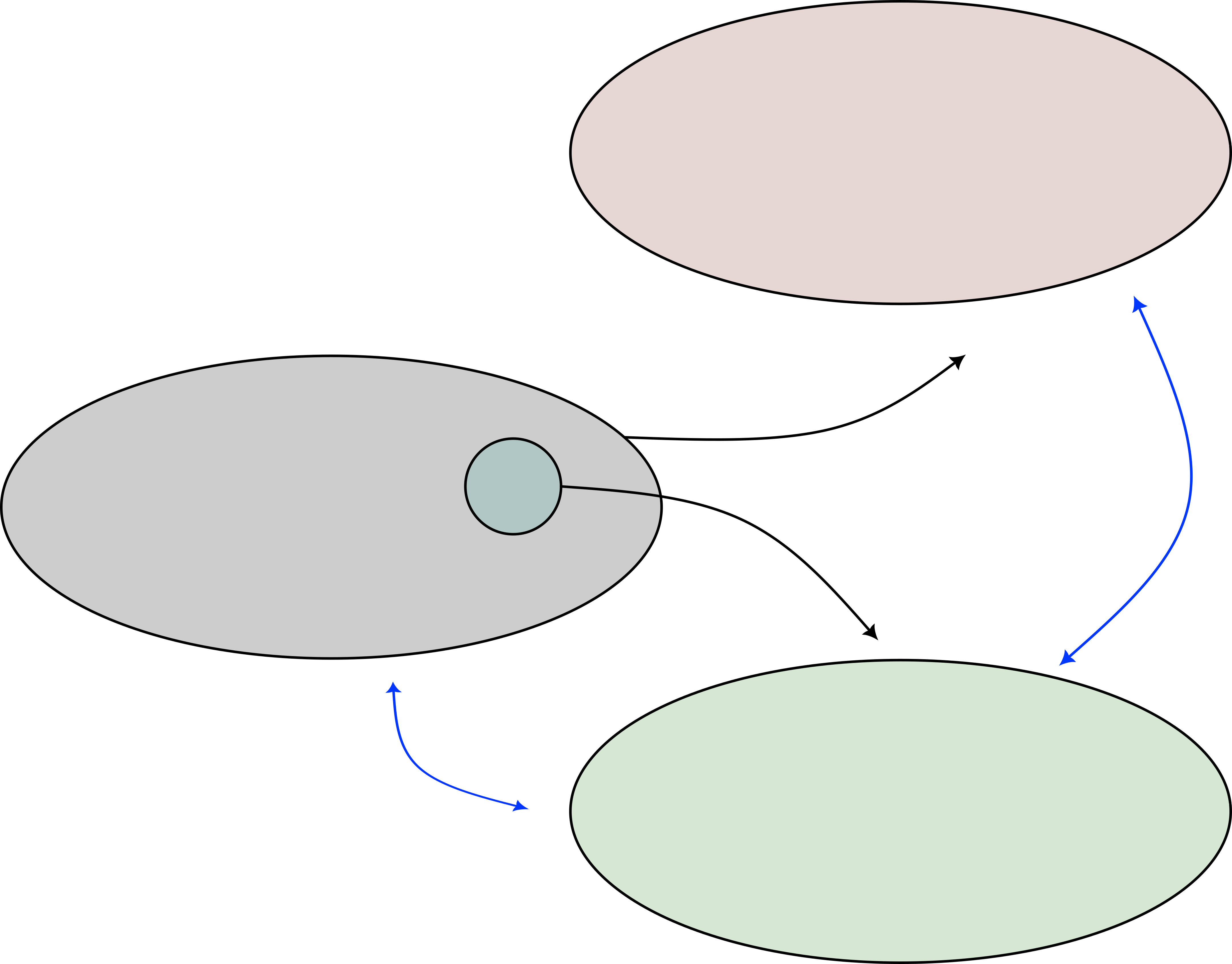}
    \caption{A visual representation of the $k$-minimizing objects in gray, containing the object $b$. The set of objects that beat a $k$-minimizing object is green, and the set of objects that lose to a $k$-minimizing object is red. It is clear that the green and grey sets weakly dominate the red set.}
    \label{fig:krpsmin}
\end{figure}
 
 Then, to see that this is unplayable, consider the sub-game with all minimizing objects, $M_k$, and objects that beat said objects, $B_k$. Any object, $i$, not in this sub-game is weakly dominated by $b$, i.e., $i$ loses to all objects in $M_k$ and, in the best case, beats all objects in $B_k$ and all other objects outside the subgame. However, by assumption, $b$ beats all of $B_k$, the objects outside the subgame, and $i$. Therefore, $i$ is weakly dominated, and the game is not playable.
\end{proof}

We can use lemma \ref{kmin} to prove the corollary \ref{mincor}. In the case where each $k$-minimizing object loses to exactly one other non-minimizing object, this would make the sum:
 \begin{equation}\label{RPS_bound_small}\sum_{i\in M_k} e_{out,i}\geq \frac{k(k-1)}{2}+k\end{equation}
 Where the first term is the contribution to $e_{out,i}$ between the subgraph of minimizing objects, and the second term is the contribution from the extra non-minimizing object, which beats the associated minimizing object. In the case that the set of objects that beat the $k$-minimizing objects is the rest of the entire game, in the infinite game, this sum must be infinite, but in the finite game, this sum must be at least:
 \begin{equation}\label{RPS_bound_large}\frac{k(k-1)}{2} +((2n+1)-k)\end{equation}
The first term is the contribution from the $e_{out,i}$ between the subgraph of $k$-minimizing objects, and the second term is the minimum number of edges to the rest of the objects in the game. This bound is smaller than the above bound \eqref{RPS_bound_small} when $k=n+1$, where the minimum is $\frac{(n+1)n}{2}+n$. By a similar argument, reversing the direction of each domination, the same is true for $e_{in,i}$. This finishes the proof of the corollary.

In the given imbalanced $(2n+1)$-$RPS$ constructions the $k$-minimizing set is the set of $\{r_i|i<k\}$, and $\{r_i| i\leq n\}\cup \{s\}$, and the $k$-maximizing set is the set of $\{p_i|i<k\}$, and $\{p_i| i\leq n\}\cup \{s\}$. As the provided imbalanced game is playable and has sequences of 
\begin{equation}
    \mathbf{e}_{out}^\uparrow=\{1,2,\dots,n,n\dots\}, \text{ and } \mathbf{e}_{in}^\uparrow= \{1,2,\dots,n,n\dots\}
\end{equation}
This sequence reaches these bounds. Therefore, these bounds must be strict, and by lemma \ref{minismaj} this game's $\mathbf{e}_{in}$ majorizes that of all other playable $(2n+1)$-$RPS$ games.

Therefore, as this sequence of $\mathbf{e}_{in}$ defines the tournament among playable tournaments, the above construction must be the unique maximally imbalanced game for all imbalance measures in the strict $S_u$ class and a maximally imbalanced game for all imbalance measures in the $S_u$ class. 

\subsubsection*{$\N$-$RPS$ Combinatorial Imbalance}
For the uniform imbalance in $\N$-$RPS$ we can see that our construction has no object $i$ with infinite $e_{min,i}$. Moreover, the minimal possible sequence of $a_{min,m}=\sum^m_{i=0}e_{min,i}^\uparrow$ is also given by:
\begin{equation}
    a_{min,m}= min(\{a_{out,k}+a_{in,m-k}| 0<k<m\})
\end{equation}
 Where, among any $\N$-$RPS$ games:
\begin{equation}
    a_{out,k}=\sum^k_{i=0}e_{out,i}^\uparrow
    ,a_{in,k}=\sum^k_{i=0}e_{in,i}^\uparrow
\end{equation}
In the constructed imbalanced $\N$-$RPS$ game, these minimum $a_{out,k}$ and $a_{in,k}$'s are reached as mentioned above, and so we must have that each $a_{min,i}$ is minimal over any other playable $\N$-$RPS$. Thus, this construction maximizes the uniform imbalance over $\N$-$RPS$s. Moreover, this minimal sequence of $\mathbf{e}_{min}^\uparrow=(1,1,2,2,3,4,4,5,5,6,6\dots)$, as well as the lack of infinite elements uniquely determines the imbalanced $\N$-$RPS$ construction. We can see this as the only playable $\N$-$RPS$ with this sequence of $\mathbf{e}_{min}^\uparrow$, and no vertex with infinite $e_{min,i}$, must have:
\begin{equation}
    \mathbf{e}_{in}^\uparrow=\{1,2,\dots\},
    \mathbf{e}_{out}^\uparrow=\{1,2,\dots\}
\end{equation}
As these two sequences account for each vertex in the tournament, we can uniquely describe the playable $\N$-$RPS$ with these sequences, by following the logic of our above construction, e.g., the object $i$ with $e_{out,i}=1$ must lose to an object $j$ with $e_{in,j}=1$, etc.

\subsubsection*{Maximal Majorization of Nash Probabilities over Playable Games}
 
Above, the probability of choosing each object $r_l,p_l$ was shown to be $\frac{1}{3^{l}}$. We can see that the maximum probability of choosing any object in a Nash equilibrium of a playable $RPS$ is $\frac{1}{3}$. For if object $a$ has probability of being played $P(a)=\frac{1}{3}+\epsilon$ then the objects which beat $a$, $\{a_w\}$, and the objects which lose to it, $\{a_l\}$, must have the same probability of being played $P(\cup_w a_w)=P(\cup_l a_l)=\frac{1}{3}-\frac{\epsilon}{2}$. Then, playing in the subgame containing objects $\{a_w\}$ must have a positive expected payoff. For even if all of $\{a_l\}$ beat all of $\{a_w\}$, the expected payoff averaged over all objects in $\{a_w\}$ receives $0$ expected payoff from the other elements in $\{a_w\}$, $-(\frac{1}{3}-\frac{\epsilon}{2})$ expected payoff from the objects in $\{a_l\}$ and $\frac{1}{3}+\epsilon$ expected payoff from $a$, for a total expected payoff of $\frac{\epsilon}{2}$. Therefore, this cannot be a Nash equilibrium.
   
If two objects have $\frac {1}{3}$ probability of being played, then every object, $o$, in the remaining set must lose to one of them and beat the other, as otherwise, $o$ would have a non-zero expected payoff. Therefore, if there is more than one remaining object, in the subgame containing the remaining objects, repeating the above argument works to show that the maximum probability of any remaining object is $\frac{1}{9}$. Inductively, the sequence of majorized probabilities is then:

\begin{equation}
    \{\frac{1}{3},\frac{1}{3},\dots \frac{1}{3^i},\frac{1}{3^i}, \dots \frac{1}{3^{(n)}},\frac{1}{3^{(n)}},\frac{1}{3^{(n)}}\}
\end{equation}
In the infinite case, it is:
\begin{equation}
    \{\frac{1}{3},\frac{1}{3},\dots \frac{1}{3^i},\frac{1}{3^i},\dots\}
\end{equation}
This finishes the proof that the given imbalanced games, both finite and infinite, uniquely maximize all imbalanced definitions in the strict $S_N$ class, e.g., distribution/$N_t$ imbalance, and are a maximally imbalanced $RPS$ with respect to any imbalanced definition in the  $S_N$ class, e.g., $N_e$ imbalance. 

\subsubsection*{Maximizing $UI_e$}

Maximizing $UI_e$, by definition, corresponds to maximizing the entropy of the uniform expected payoff distribution. Therefore, for achievable scores $\frac{i}{2n}$ for $i\in \Z$, it corresponds to maximizing:
\begin{equation}
    H(P)=-\sum_{i\in \Z} P(score=i/2n)\ln( P(score=i/2n))
\end{equation}

We can see that generally, $H$ is maximized when each $P(score=i/2n)$ that is non-zero is minimized. As the only possible values of $e_{in,i}-e_{out,i}$ in a playable $(2n+1)$-$RPS$ are even and between $2n-2$ and $-2n+2$, we have $2n-1$ score values and $2n+1$ objects with scores. Therefore, we cannot uniformly place the minimum probability of $\frac{1}{2n+1}$ at each possible score value. 

The best we can hope to achieve is that $2n-3$ scores have the minimum probability of $\frac{1}{2n+1}$, and two scores have a probability of $\frac{2}{2n+1}$. As the average of these scores must be zero, these scores must be $-a/2n$ and $a/2n$ for some $a\neq 0$. However, by the corollary \ref{kmin}, this resulting game would not be playable as the outgoing edges from the $(2n-a+1)$-minimizing objects would be too few.

The next remaining option that still includes giving every possible score a positive probability is to have $2n-2$ scores have the minimum possible probability of $\frac{1}{2n+1}$ and the score of $0$ have the probability of $\frac{3}{2n+1}$. This result is achieved in the above construction. The entropy of this result is:

\begin{equation}
\begin{split}
    -\left(\frac{2n-2}{2n+1}\ln\left(\frac{1}{2n+1}\right)+ \frac{3}{2n+1}\ln\left(\frac{3}{2n+1}\right)\right)=\\
     \ln(2n+1)-\frac{3}{2n+1}\ln(3)
\end{split}
\end{equation}

The next best option would not include giving every possible score a positive probability. Instead it would involve giving $2n-5$ scores a probability of $\frac{1}{2n+1}$ and three scores a probability of $\frac{2}{2n+1}$. The entropy of this outcome is:

\begin{equation}
\begin{split}
    -\left(\frac{2n-5}{2n+1}\ln\left(\frac{1}{2n+1}\right)+ \frac{6}{2n+1}\ln\left(\frac{2}{2n+1}\right)\right )=\\
     \ln(2n+1)-\frac{6}{2n+1}\ln(2)
\end{split}
\end{equation}
As $6\ln(2)>3\ln(3)$, the imbalanced game provided above is maximally imbalanced with respect to $UI_e$ while remaining playable.

Therefore, we have shown that both the game-theoretic and combinatorial imbalance statistics for games agree on the maximum imbalanced $RPS$ games that remain playable, thereby completing the theorem \ref{unfairthm}.

\section{Two Use Cases: Ecological Modeling and Competitive Card Games} \label{applications}

\subsection{An Ecological Game}
As described above, we can model population dynamics with an $n$-$RPS$, $G$, where each subgroup considered corresponds to an object or pure strategy in $G$, and the direction of each edge indicates which group benefits from the increased population of the other, different forms of this have been considered in \cite{LairdSchamp2015IntransitivityMEE}, \cite{LairdSchamp2018ComputationalChallenges}, and  \cite{KerrEtAl2002LocalDispersalRPS}. We can then model the changes of an ecosystem as a sequence of games of $G$. More precisely, we can consider the extended game, $\mathfrak{G}$, where we partition a time period $t=\{t_1\dots t_i\dots \}$, and have a sequence of players, $p_i$, each playing the given game (an $n$-$RPS$, for example), $G$, with the previous player, $p_{i-1}$, and the later player $p_{i+1}$. Each player is restricted to playing the same mixed strategy for both games, and a player at a later time knows the mixed strategy of the player at the previous time. In this analogy, the mixed strategy of $p_i$ corresponds in some way to the population distribution at time $t_i$.

For instance, at time $t_{i-1}$, there might be an overabundance of bunnies, corresponding to the player, $p_{i-1}$, playing a mixed strategy with too high a probability of playing the bunny object. Then, at time $t_{i}$, there may be many wolves that consume all the bunnies. In terms of the game, this corresponds to the player at time $p_{i}$ choosing to maximize his payoff against $p_{i-1}$ by playing a higher proportion of wolves. It may be that, as the wolves kill most small animals, the flora will become abundant in the coming years, while the fauna will dwindle. In the game, this corresponds to the player $p_{i+1}$ playing no bunnies and only flora to maximize his potential winnings against player $p_{i}$.

We can show that each Nash equilibrium of $G$ defines a Nash equilibrium for $\mathfrak{G}$. As for an equilibrium, $(s_1,s_2)$, if player $p_1$ plays mixed strategy $s_1$, player $p_2$ cannot get a better payoff against $p_1$ than that given by playing $s_2$. $p_2$ similarly knows that if he plays a non-equilibrium strategy, $s$, $p_3$ could take advantage for a net benefit, and as $G$ is zero-sum, $p_2$ has a strictly negative payoff. Therefore, $p_2$ must also choose a strategy, ${s_2}'$, such that $(s_1,{s_2}')$ is a Nash equilibrium for maximizing his expected payoff against both $p_1$ and $p_3$. Inductively, we can see that this continues for all players. Note that not all equilibrium positions must come from the set of equilibria of the underlying $RPS$ game. However, just as in the case of the side-blotched lizards, we can see that there may be stable equilibria of $\mathfrak{G}$ that seem to rotate around the equilibria of $G$.

This model can also be adjusted to account for environmental changes and the instability they cause. For example, say that there was a massive fire at time $t_i$ killing all wolves; this could correspond to $p_i$ being unable to choose the object corresponding to wolves in his mixed strategy. In a playable game, this may cause later players to take advantage of the suboptimal position of previous players. For example, in standard $RPS$ terminology, we can consider what the optimal strategy is for $p_i$ given that they cannot play paper. Moreover, how would this change $p_{i+1} $'s strategy? It would also be interesting to know how the underlying game affects whether the dynamic restabilizes, and if so, how long that would take.

\begin{rems}[Directionality of $\mathfrak{G}$]
    In $\mathfrak{G}$ each player $p_i$ played against the previous player $p_{i-1}$ and the next player $p_{i+1}$; however, in reality we should expect the current population distribution to be affected substantially by the previous population distribution and to a much lesser degree by the influence of a future population distribution, for instance through attractors \cite{TaylorJonker1978Replicator}. We could account for this by weighing the payoff for $p_i$ in the later game with $p_{i+1}$ by some constant $0<\lambda<1$. 
\end{rems}

In this model, a playable game implies the existence of a stable system in which each species in the system has a positive existing population, and a strongly playable game would correspond to an environment where every stable equilibrium must have the property that every considered group has a positive population. For example, we can use this model as a base to describe Darwin's finches fairly accurately. In the stable equilibrium of each island, certain finches have positive populations; However, there are few, if any, islands whose stable equilibrium contains all species of finch \cite{GrantGrant2008}. The Galapagos ecosystem would correspond to a game $\mathfrak{G}$ built upon a game $G$ that is weakly playable, but not necessarily playable, as each finch has a Nash equilibrium where it sees play, but no Nash equilibrium plays all finches.

In contrast, an imbalanced base game $G$ corresponds to an ecological game $\mathfrak{G}$ where a particular group dominates over many others, and or certain groups seem to be dominated by many others. For instance, consider the following hypothetical sketch of a food web of the coniferous forest:
\begin{ex}[Imbalanced Ecosystem]
    Consider an ecosystem containing:
    \begin{itemize}
        \item A species of wolf that dominates nearly all other fauna. 
        \item Several species of large fauna, e.g., elk, moose, that dominate smaller fauna and most flora.
        \item Several species of small fauna, e.g., rodents, etc., which dominate nearly all types of flora.
        \item Several types of flora, e.g., some reachable by large fauna, and some only reachable by small fauna.
    \end{itemize}
\end{ex}
This ecological system is imbalanced in that the wolf seems to dominate over all other fauna in the region, and the flora is dominated by nearly all the fauna in the region. Despite this imbalance, if the underlying game were strongly playable, that would imply that without positive populations for the different species in the food web, the food web collapses. We conjecture that this asymmetric yet stable scenario can naturally apply evolutionary pressure to a subspecies while maintaining a given niche for each species to avoid complete species extinction. The motivation for this approach is to combat the imbalance by forming new niches. This behavior is alluded to in the fact that as the number of strategies in the maximally imbalanced $RPS$ increases, the maximal imbalance statistics decrease.

To try to extend these concepts to model genuine ecological relationships better, one would have to extend past simple $RPS$ games. Extensions to more complicated games would allow us to weigh each species' interaction, rather than just categorizing it as advantageous or disadvantageous, and also enable us to account for symbiotic relationships between different species.

\subsection{Imbalance and Playability as Health Indicators in Card Games}
We can also apply the above ideas to competitions in trading card games such as \textit{Magic: The Gathering} or \textit{Yu-Gi-Oh}. In these competitions, each player chooses a deck and plays several games against other players. In high-level play, each deck can be categorized into one of several archetypes. In terms of an underlying $RPS$, each archetype can be considered an object that has an advantage against certain archetypes and a disadvantage against others. For instance, as we described above, decks that start very aggressively generally lose against decks that build up a little slower and more linearly, which typically lose against even slower combo decks, which then, in turn, lose against the very fast aggressive decks. 

Different pure strategies in a game $G$ can then serve as a model for the archetypes played between two random competitive players. The playable condition on $G$ then corresponds to the idea that all archetypes in the game should see play in high-level competition. Moreover, strong playability implies that each archetype must be played, given enough players. In contrast, weak playability suggests that for each archetype, there is a meta-game in which it must be played if given enough players. This condition ensures a diverse and exciting competitive environment, as one, a priori, should expect to face several different archetypes.

Just as in the ecological game, we expect that if the game $G$ that models the meta-game is imbalanced yet playable, the meta-game supports the creation of new strategies. The fundamental motivation is the same, in that with more strategies, the net imbalance seems to decrease. For instance, suppose an imbalanced but strongly playable game $G$ was a model for the current interactions between archetypes. In that case, both the archetypes that dominate most archetypes and the archetypes that are dominated by most archetypes must see play. The imbalance/asymmetry within $G$ makes brewing new archetypes and sub-archetypes very appealing, as there is a great deal of creative freedom in attempting to dominate different sets of asymmetric archetypes, while knowingly paying the cost of losing regularly to some other subset. For instance, to create a competitive archetype, one can begin by trying to dominate the current highest dominating archetype. In so doing, they can be sure that their deck is at least minimally competitive in that it beats a very dominant and commonly played deck. This directly adds pressure to create more balance; however, creating a foil to a single archetype, while generally easy, will not be very competitive in random matchmaking. This reality furthers the drive to compete against multiple sets of archetypes, further cementing a new niche in the meta-game if successful. In this way, imbalance can further add evolutionary pressure, while playability provides the stability that ensures that competition stays diverse.

In the case of competitive \textit{Magic: The Gathering}, this behavior is highly sought after. In particular, in past healthy equilibria, certain well-known archetypes have been known to beat most other archetypes, yet consistently make only one or two spots in the top eight in large competitions. On the other hand, in past unhealthy equilibria, single archetypes have occupied six of the top eight spots, and the game as a whole was considered by many to be unplayable \cite{mtghist}. We could further expand the validity of this model by using a more complex game that accounts for the variance in win probabilities between different archetypes.

\section{Appendix: Blow-Ups in Games}

    In the following section, we develop a construction for combining two games by analogy to the topological concept of blow-ups. In the sequel, we extend this construction to multiplayer $3$-$RPS$ games as a way to construct playable $m$-player $(2n+1)$-object $RPS$, which we conjecture to be maximally imbalanced. 
    The topological concept of a blow-up replaces a point in a manifold with a full copy of a different canonical manifold, in a way that stitches together the relevant structures nicely. Our new construction is analogous in that, in the symmetric case, we replace an object in the game with an entirely new game that preserves as much structure as possible from the original. 
    
    In terms of graph theory, this can be thought of as a modular product of the two graphs corresponding to the given games, as in \cite{Salha2022MaxDecomposability}. Similarly, in game theory, this can also be considered as the partial lexicographic product of two games, where again the preference between the first and the second game only occurs when tied on a singular pure strategy, as opposed to throughout the existence of both games. 
    
    We can more precisely define a blow-up here as:    
    \begin{defi}[Blow-Ups]
        The blow-up of two discrete two-player games, $G_1, G_2$, over pure strategies $l\in G_1$ for player $p$ and $m\in G_1$ for player $q$, denoted as $G_1\#_{(l,m)} G_2$, is a game with the pure strategies in $(G_1\setminus \{l,m\}) \cup G_2$. If we let $(G_1\setminus (l,m))$ be the payoff matrix of $G_1$ without the row corresponding to $l$ and the column corresponding to $m$, $(G_1)|_l$ be the $(|G_1|-1)\times |G_2|$ matrix with each column being a copy of the column corresponding to $l$ in the payoff matrix of $G_1$, and similarly define $(G_1^\top)|_m$ as the $|G_2|\times (|G_1|-1)$ matrix with each row being a copy of the row corresponding to $m$ in the payoff matrix of $G_1$, then the payoff matrix of this new game is:
        \begin{equation}
            A_{G_1\#_{(l,m)} G_2}=\begin{bmatrix}
                (G_1\setminus (l,m)) & (G_1)|_l \\
                (G_1^\top)|_m& G_2 
            \end{bmatrix}
        \end{equation}
        In this new game, if both players play strategies in $G_2$, the payoff matrix is the payoff matrix in $G_2$. If both players play in $G_1$, with neither player playing $l$ or $m$, respectively, then the payoff matrix is the payoff matrix in $G_1$. If player $p$ plays $j\in G_2$, and $q$ plays $i\in G_1\setminus m$ the payoff is the same as those in $G_1$ if $p$ played $l$ and $q$ played $i$. Similarly if $q$ played $k\in G_2$, and $p$ played $i\in G_1\setminus l$ the payoff is the same as those in $G_1$ if $q$ played $m$ and $p$ played $i$. For symmetric games $G_1\#_l G_2=G_1\#_{(l,l)}G_2$.  
    \end{defi} 

    \begin{figure}[ht]
        \centering
        \includegraphics[width=.5\linewidth]{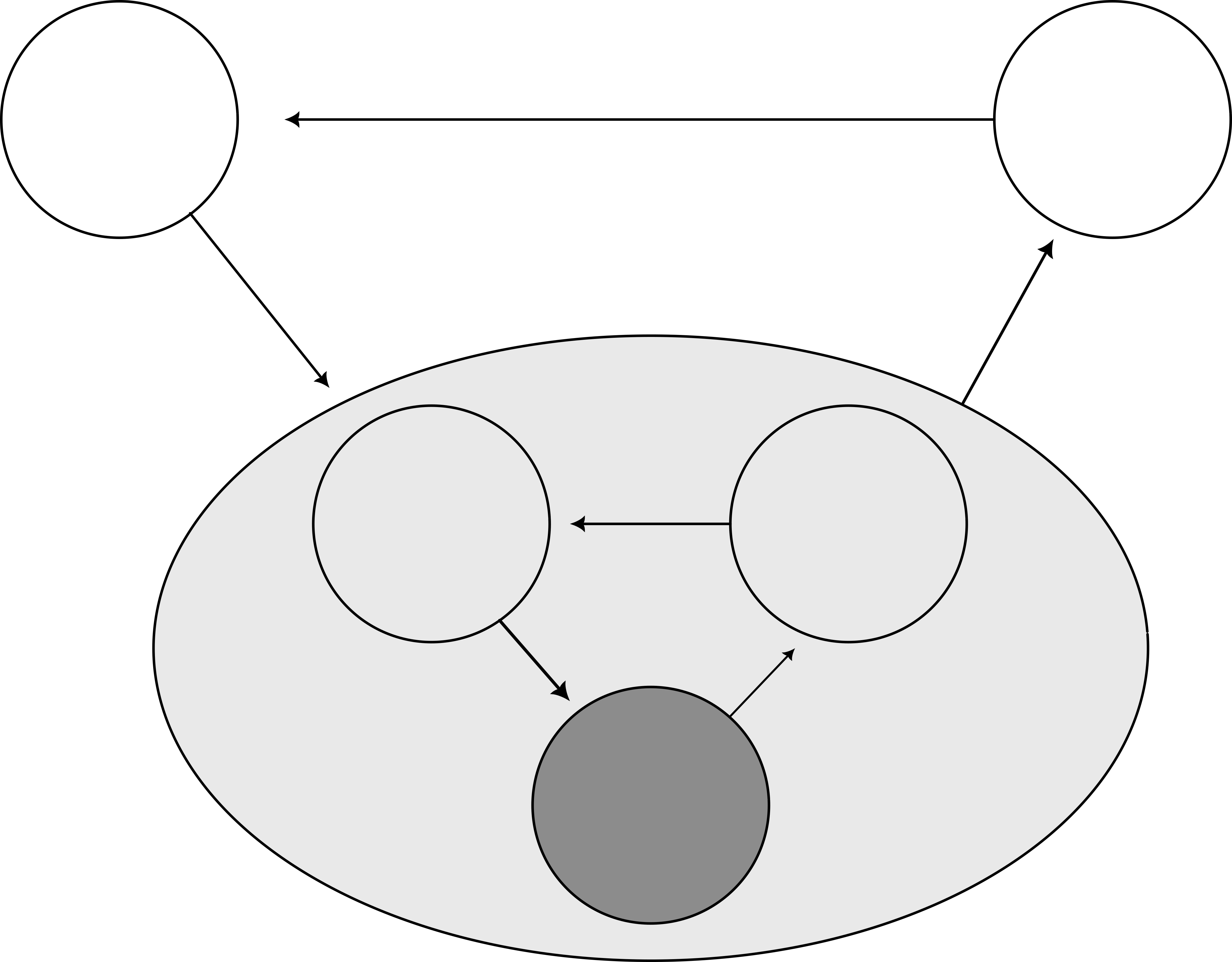}
        \caption{A visual representation of blowing up in $3$-$RPS$. The light gray circle is a blow-up at $s$ in the top $3$-$RPS$. We can again blow up with this game in the dark gray circle to continue this process recursively.}
        \label{fig:enter-label}
    \end{figure}
    
A blow-up of a playable game by another playable game is not necessarily playable. However, in the symmetric zero-sum case, given any playable Nash equilibrium, for $G_1$, $(\mathbf{v}_1,\mathbf{v}_2)$, and $G_2$, $(\mathbf{w}_1,\mathbf{w}_2)$, we can construct a Nash equilibrium in $G_1\#_{l}G_2$, with probability vectors:
    \begin{equation}
    \begin{split}
        \{(v_{1,1},\dots, v_{1,l-1},v_{1,l+1},\dots, v_{1,n},(v_{1,l})w_{1,1},\dots, (v_{1,l})w_{1,m}), \\
    (v_{2,1},\dots, v_{2,l-1},v_{2,l+1},\dots,v_{2,n}, (v_{2,l})w_{2,1},\dots, (v_{2,l})w_{2,m})\}
    \end{split}
    \end{equation}
     If this were not a Nash equilibrium, then one of the players could get an advantage by changing strategies; however, we can project $p$'s strategy in $G_1\#_{l}G_2$ to $G_1$, which takes it to $\mathbf{v}_1$ or project it to $G_2$, taking it to $\mathbf{w}_1$. By applying the corresponding payoff matrices, in $G_1, G_2$, $q$ has a non-positive expected value for each object; by iterated expected values, the same is true in $G_1\#_{l}G_2$. Therefore, in the blown-up game, this is a Nash equilibrium.
    
    If $G_1$ were 2-strongly playable, then every Nash equilibrium would be in this form. We can demonstrate this by noting that all players must play in the copy of $G_2$ in  $G_1 \#_{l} G_2$ in every Nash equilibrium. As if player $p$ does not play in $G_2\subset G_1\#_{l}G_2$, their payoff is identical to that of if they chose a mixed strategy in $G_1$ that does not involve object $l$ in $G_1$. As $G_1$ is 2-strongly playable, this must give $q$ a strategy with a net positive payoff, $s^q$, in $G_1$. Applying $s^q_o$ for all objects $o\in G_1\setminus l \subset G_1\#_{l}G_2$, and $s^q_l$ to some object in $G_2 \subset G_1\#_{l}G_2$, exactly mirrors the payoffs for $s_q$ as a strategy in $G_1$. Thus, as both players must play strategies in $G_1\#_{l}G_2$ that canonically project to $G_1$ and $G_2$, by iterated expected value, the projection of the strategies is also a Nash equilibrium strategy in $G_1$ and $G_2$.
    
   The purpose of highlighting this construction of blow-ups is that our maximally imbalanced balanced $(2n+1)$-$RPS$ can be equivalently defined as: 
   $3$-$RPS\#_{S}(3$-$RPS\#_{S}\dots)$ $n$ times, and $\N$-$RPS$ is defined by infinite blow-ups: $3$-$RPS\#_{S}(3$-$RPS\#_{S}(\dots$. This construction reveals that the least balanced but playable $(2n+2m+1)$-$RPS$ can be constructed by blowing up the least balanced but playable $(2n+1)$-$RPS$, $G_1$, with the least balanced but playable $(2m+1)$-$RPS$, $G_2$, at the most balanced object, $s\in G_1$. This leads to the conjecture:
   \begin{conj}[$m$-player Imbalanced $RPS$ Games come from Blow-ups]
       The most imbalanced yet strongly playable $m$-player $(2n+1)$-$RPS$ game is equivalent to the blow up of the most imbalanced yet strongly playable $(2m+1)$-$RPS$ game, $G_1$ with the most imbalanced yet strongly playable $(2(n-m)+1)$-$RPS$ game, $G_2$ at the most balanced object of $G_1$, for an $m<n$.
   \end{conj}
   We show evidence for this conjecture in the sequel.

\phantomsection
\bibliographystyle{amsplain}
\bibliography{references}

\end{document}